\documentclass[11pt,draftcls,onecolumn]{IEEEtran}
\usepackage[comma,numbers,square,sort&compress]{natbib}
\usepackage{amsthm}
\usepackage{mathrsfs}
\usepackage{array}
\usepackage{amssymb}
\usepackage[cmex10]{amsmath}
\usepackage{float}
\usepackage{graphicx}
\usepackage{color}
\DeclareMathOperator{\diag}{diag}
\DeclareMathOperator{\E}{E}
\DeclareMathOperator{\rank}{rank}
\newcommand{\set}[1]{\left\{#1\right\}}
\newtheorem{definition}{Definition}

\newtheorem{theorem}{Theorem}
\newtheorem{remark}{Remark}
\newtheorem{lemma}{Lemma}

\begin{document}

\title{Local Mode Dependent Decentralized $H_{\infty}$ Control of Uncertain Markovian Jump Large-scale Systems \thanks{This work was supported by National Natural Science Foundation
    of China under Grant 61004044, Program for New Century Excellent Talents
    in University (11-0880), and the Fundamental Research Funds for the
    Central Universities (WK2100100013).}
}
\author{Guilin Zhuansun, Junlin Xiong\thanks{Department of Automation, University of Science and Technology of
    China, Hefei 230026, China.}}

\maketitle

\begin{abstract}
This paper considers the problem of robust $H_{\infty}$  control using decentralized state feedback controllers for a class of  large-scale systems with Markov jump parameters.  A  sufficient condition is developed to design   controllers using local system states and local system operation modes. The sufficient condition is given in terms of rank constrained linear matrix inequalities. An illustrative numerical example is given to demonstrate the developed theory.
\end{abstract}

\section{Introduction}
In physical world, many systems, for example, power systems, digital communication networks, economic systems and ecological systems, can be seen as large-scale systems. Generally, large-scale systems represent a class of
dynamic systems with high dimensions and complex structures, and such systems are usually  characterized by geographical separation. When economic cost and reliability of  communication links have to been taken into account,
the decentralized control scheme is useful~\cite{Sandell1978}. In recent control literatures, much attention has been given to decentralized control problems~\cite{Ugrinovskii2005,Xiong2009,Ugrinovskii2000}. Compared with centralized controllers, decentralized controllers are designed only using locally available information of the subsystems, which means a lower level of connectivity and  communication costs.

 During the past years, there has been quite a lot of attention devoted to the study of  Markovian jump linear systems~(MJLSs)~and some useful control theories for this class of systems have been developed, such as stabilization~\cite{Farias1992,Ma2009,Feng1992}, $H_2$ control~\cite{xiong2009},  $H_{\infty}$ control~\cite{Fagoso1993,Li2007,Zhang2009}, 2D systems~\cite{Wu2012H2,Wu2008}, singular systems~\cite{Wu2012L2,Wu2010}, and model reduction~\cite{Zhang2008}. Factually, large-scale systems are often subject to abrupt  changes in their structures due to component failures or repairs, sudden environmental disturbances, changing subsystem interconnections, etc. These structural variations often lead to random variations of the system parameters. In these situations, the model of MJLSs can be used  to capture the abrupt changes and some study of large-scale systems  using the model of MJLSs can be seen in ~\cite{Ugrinovskii2005,Xiong2009,Xiong2010,li2007}.  In the system setup, the changes of operation modes in the large scale system are governed by a
Markov process, referred to as  global mode process, while the subsystems of the large scale system have their own operation modes.
 Due to the interaction of the subsystems, the changes of subsystem operation modes are not necessarily Markovian.  In ~\cite{Xiong2009,Xiong2010}, the controllers  designed only using the state and mode information locally available within the subsystems are referred to as local mode dependent controllers. A local mode dependent controller changes its operation mode only when the corresponding subsystem changes its operation mode. The  controllers studied in~\cite{Ugrinovskii2005,li2007}, known as global mode dependent controllers, change their gains whenever any of the subsystems switches to a different mode.  Hence, compared with the global mode control scheme, the local mode control scheme can effectively reduce the number of control gains  and remove undesirable transient dynamics.   As  the local  mode processes for the subsystems is non-Markovian, the local mode dependent controllers can not be obtained directly, but can be derived from the global ones, and the derivation can be seen in Section \uppercase\expandafter{\romannumeral3}-B of \cite{Xiong2009}.

 In the large-scale systems considered in this paper, there are two kinds of internal uncertainties  affecting the dynamics of system considered in this paper, which are local uncertainties and interconnection uncertainties. The local uncertainties  result from the existence of uncertain dynamics in each subsystem and the interconnection uncertainties  result from the fact that the subsystems are interrelated and interact with each other. In this paper these uncertainties are assumed to satisfy integral quadratic constraints (IQCs), which have been proven to be effective in a number of robust control problems, for example, see~\cite{Ugrinovskii2005,Petersen2006,li2007}.

 In \cite{Xiong2009}, with no exogenous disturbances entering the system, the authors considered the problem of decentralized stabilization and developed a sufficient condition to design local mode decentralized stabilizing controllers. However, the external
disturbances are often unavoidable in more practical situations. Apart from internal uncertainties, each subsystem of the large-scale system is also affected by exogenous input signals.  This paper focuses on the problem of designing decentralized state feedback control laws to reduce the affect the disturbance input has on the controlled output in an $H_{\infty}$ sense, with the controllers designed using local system states and operation modes information.  The purpose is to find controllers such that  the closed-loop large-scale system is  mean square stable (MSS) and  its $H_{\infty}$ norm is below a prescribed level. This study could be considered as a further development of the result of~\cite{Xiong2009}.

This paper is organized as follows. Section \uppercase\expandafter{\romannumeral2} gives some preliminary knowledge of the Markov jump systems. Section \uppercase\expandafter{\romannumeral3} formulates the class of uncertain large-scale systems with Markovian jumping parameters.  In section \uppercase\expandafter{\romannumeral4}, a sufficient condition is developed to construct local mode dependent decentralized robust $H_\infty$ controllers so that the closed-loop system is mean square stable  and its $H_{\infty}$ norm is below a prescribed level.
 A numerical example is presented in section \uppercase\expandafter{\romannumeral5} to illustrate the usefulness of the developed theory. Section \uppercase\expandafter{\romannumeral6} concludes this paper.

$Notation$: $\mathbb{R}^{+}$ denotes the set of positive real numbers, $\mathbb{S}^{+}$ denotes the set of real symmetric positive definite matrices. The $\star$ in symmetric matrix is used to represent the blocks induced by symmetry.   $ \|\cdot\|$ refers to the Euclidean norm for vectors and the induced 2-norm for matrices.  $L_2[0,\infty)$ denotes the space of square integrable vector functions  over $[0,\infty)$.
\section{Preliminaries}
Consider a continuous-time  Markov jump parameter system $\mathcal{S}_0$ on a complete probability space $(\Omega,\mathcal{F},\Pr)$ described by the following state-space equations:
\begin{align*}
\mathcal{S}_0:\left\{
\begin{array}{l}
\dot{x}(t)=A(\eta_t)x(t)+B(\eta_t)w(t),\\
z(t)=C(\eta_t)x(t)+D(\eta_t)w(t),
\end{array}
\right.
\end{align*}
where $x(t)\in\mathbb{R}^n$ is the state, $w(t)$ is the disturbance input, assumed to be an arbitrary signal in $L_2[0,\infty)$, $z(t)$ is the controlled output,  $\eta_t$ is a time-homogeneous Markov process  taking values on the finite set $\mathcal{M}=\set{1,2,\ldots,M}$. Let $Q=(q_{\mu\nu})\in\mathbb{R}^{M\times M}$ be the transition rate matrix of the process, with $q_{\mu\nu}\geqslant 0,~\mu\neq \nu$ and $q_{\mu\mu}=-\sum_{\nu\neq \mu}^{M} q_{\mu\nu}\leqslant 0$, such that transition probabilities of the system mode variable $\eta_t$ satisfy
\begin{equation*}
P(\eta_{t+\Delta} =\nu |\eta_t=\mu)=\left\{
\begin{array}{l l}
q_{\mu\nu}\Delta+o(\Delta), &\mu\neq \nu,\\
1+q_{\mu\mu}\Delta+o(\Delta), &\mu =\nu,
\end{array}
\right.
\end{equation*}
 When $\eta_t=i$, we have $A(\eta_t)=A_i$, $B(\eta_t)=B_i$, $C(\eta_t)=C_i$, $D(\eta_t)=D_i$, respectively. The following definition generalizes the concept of   mean square stable for Markov jump systems~\cite{Farias1992}.

\begin{definition}
The system $\mathcal{S}_0$ with $w(t) \equiv 0$ is mean square stable  if
\begin{equation*}
\lim_{t\rightarrow \infty}\E(\|x(t)\|^2)=0.
\end{equation*}
for any initial condition $x(0)\in\mathbb{R}^n$ and  $\eta_0\in\mathcal{M}$.
\end{definition}
The next definition is a generalization of  the $H_{\infty}$-norm for  Markov jump systems~\cite{Fagoso1993}.
\begin{definition}
Suppose system $\mathcal{S}_0$ is mean square stable, then the $H_{\infty}$-norm of the system is defined as
 \begin{equation*}
\|\mathcal{S}_0\|_{\infty}=\sup_{\substack{ \ w(\cdot) \in L_2[0,\infty) \\x(0)=0, \|w(\cdot)\|\neq 0}} \frac{\left\{\E\left[\int_0^{\infty}\|z(t)\|^2dt\right]\right\}^{1/2}}{\left\{\int_0^{\infty}\|w(t)\|^2dt\right\}^{1/2}}.
\end{equation*}

\end{definition}

With the definitions given above, an LMI characterization for the $H_{\infty}$ norm of system $\mathcal{S}_0$ can be described by the following bounded real lemma, see~\cite{Zhang2008, Marcos2008}.
\begin{lemma}
 Given $\gamma>0$, system $\mathcal{S}_0$ is mean square stable with $\|\mathcal{S}_0\|_{\infty}<\gamma$ if and only if there is $P_i \in \mathbb{S}^{+}$ such that the coupled LMIs
\begin{align*}
\begin{bmatrix}
A_i^TP_i+P_iA_i+\sum_{j=1}^M q_{ij}P_j & P_i B_i & C_i^T\\
B_i^T P_i & -\gamma^2 I & D_i^T \\
C_i & D_i & -I
\end{bmatrix}<0
\end{align*}
are satisfied for all $i=1,\ldots, M$.
\end{lemma}

\section{Problem Statement}

Consider  an uncertain Markovian jump large-scale system
$\mathcal{S}$  which is comprised of $N$ subsystems. The $i$-th subsystem $\mathcal{S}_i$, $i \in \mathcal{N}\triangleq\{1,2,\ldots,N\} $, is described by
\begin{equation*}
\label{local mode system}
 \mathcal{S}_i:\left\{
 \begin{aligned}
 &\dot x_i(t) =A_i(\eta_{i,t})x_i(t)+B_i(\eta_{i,t})u_i(t) +E_i(\eta_{i,t})r_i(t)+F_i(\eta_{i,t})\xi_i(t)+G_i(\eta_{i,t})w_i(t),\\
&z_i(t) =C_i(\eta_{i,t})x_i(t)+D_i(\eta_{i,t})u_i(t),\\
&\zeta_i(t)=H_i(\eta_{i,t})x_i(t),
 \end{aligned} \right.
\end{equation*}
where $x_i(t)\in \mathbb{R}^{n_i} $ is the state of the subsystem $ \mathcal{S}_i$, $ u_i(t)$ is the control input,
$r_i(t)$ is the interconnection input, describing the effect of the other subsystems $ \mathcal{S}_j$, $j\neq i$, on subsystem  $ \mathcal{S}_i $, $\xi_i(t)$ is the local uncertainty input, $w_i(t) \in L_2[0,\infty) $ is the  exogenous disturbance input,  $z_i(t)$ is the controlled output,
 $\zeta_i(t)$ is the uncertainty output,  $ \eta_{i,t} $  denotes the operation mode of subsystem $ \mathcal{S}_i $  and  takes values in a finite state space $ \mathcal{M}_i\triangleq\left\{1,2,\ldots,M_i\right\}$. The initial condition of subsystem $\mathcal{S}_i$ is given by $x_{i0}\in \mathbb{R}^{n_i}$ and $\eta_{i0} \in \mathcal{M}_i $. The structure of the $i$-th subsystem is shown in Fig.1.

\begin{center}
\begin{picture}(180,180)
\put(-250,-550){\resizebox{250mm}{300mm}{\includegraphics[width=0.5\textwidth]{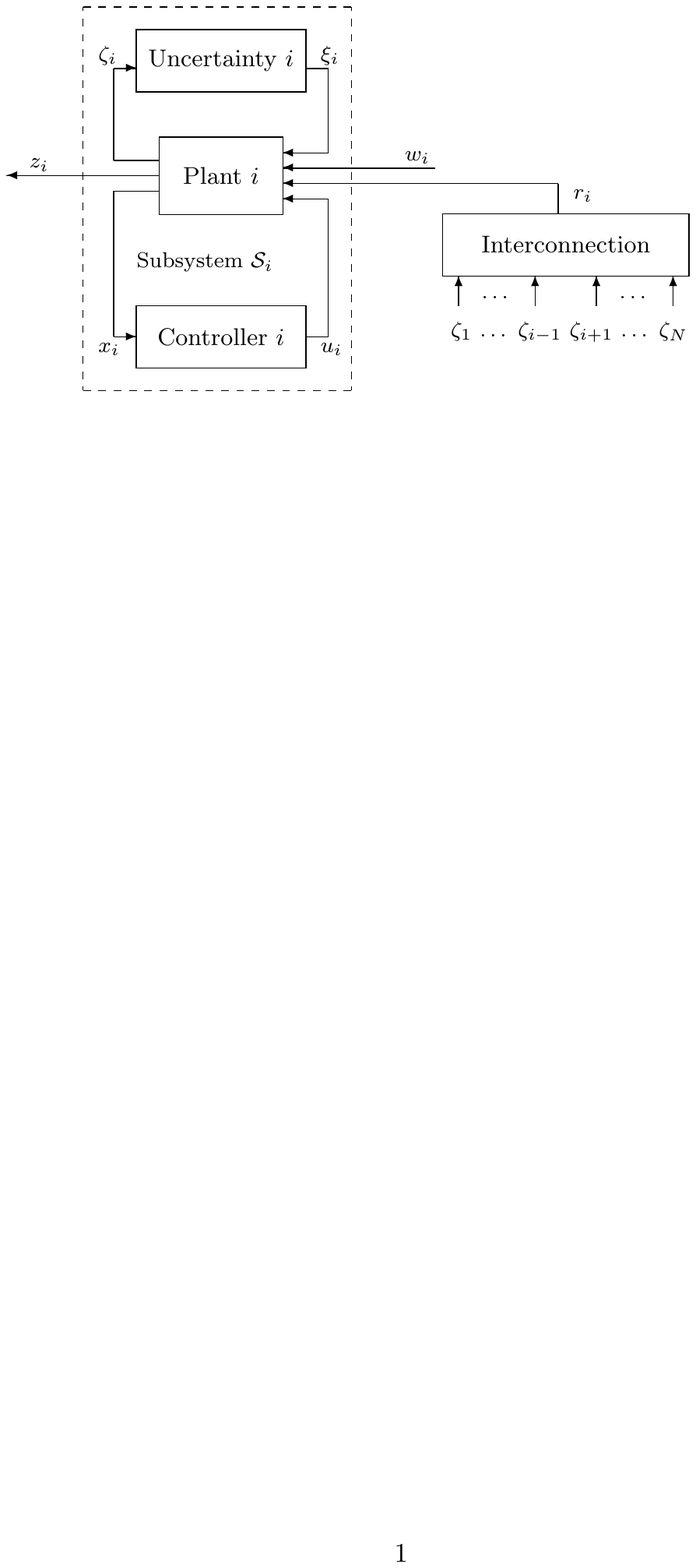}}}
\put(0,0){\shortstack{\quad \quad \quad Fig. 1. The structure of the subsystem $\mathcal{S}_i$. }}
\end{picture}
\end{center}

For the large-scale system $\mathcal{S}$,  an operation mode pattern set $\mathcal{M}_p$ is introduced to describe the vector mode states visited by $ \left[\eta_{1,t}~\ldots~\eta_{N,t}\right]$. Let $\mathcal{M}_p$ be a non-empty subset of the set $\mathcal{M}_1\times\cdots\times\mathcal{M}_N$, and is supposed to have $M$ elements where $\max_{i \in \mathcal{N}} M_i \leq M \leq \Pi_{i=1}^NM_i$. In this case we say the large-scale system has $M$ global operation modes. Let $ \mathcal{M}\triangleq\left\{1,2,\ldots,M\right\}$, then there must exist a bijective function $ \Psi : \mathcal{M}_p\rightarrow \mathcal{M} $, given by $\mu=\Psi(\mu_1,\mu_2,\ldots,\mu_N)$. The inverse function  $ (\mu_1,\mu_2,\ldots,\mu_N)=\Psi^{-1}(\mu)$ defines the mapping $ \Psi^{-1}:\mathcal{M}\rightarrow \mathcal{M}_p $. Also, the components of the inverse function $ \Psi^{-1}$ are denoted by  $ \Psi_i^{-1}$, i.e., $ \mu_i=\Psi_i^{-1}(\mu)$.  The large-scale system is said in global mode $\mu$ at time $t$ if $\Psi\left(\eta_{1,t}~\ldots~\eta_{N,t}\right)=\mu$. Defined $\eta_t\triangleq \Psi \left(\eta_{1,t}~\ldots~\eta_{N,t}\right)$, and $\eta_t$ is used to represent the global mode process. The global mode process is assumed to be a stationary ergodic Markov process with the state transition rate matrix given by $Q=(q_{\mu\nu})\in\mathbb{R}^{M\times M}$. Let $q_\infty$ denote the steady state distribution of $\eta_t$, and it can be computed according to:
$$q_\infty=\textbf{e}(Q+\textbf{E})^{-1}$$
where $\textbf{e}=[1~~ 1~\cdots~ 1]\in\mathbb{R}^{1\times M}$, $\textbf{E}=[\textbf{e}^T~~ \textbf{e}^T~\cdots~ \textbf{e}^T]^T\in\mathbb{R}^{M\times M}$

The relationship between the local operation modes of the subsystems and the global operation modes of the entire system has been established.  Using an auxiliary system, under some conditions, if the  global mode dependent controllers are given, the local ones can be derived, which will be shown in details in the next section.

It is also assumed that the system $\mathcal{S}$ satisfies the following assumptions. For all $i=1,\ldots, N$,
\begin{enumerate}
\item $D_i^T(\eta_{i,t})[C_i(\eta_{i,t})~~ D_i(\eta_{i,t})]=[0 ~~R_i(\eta_{i,t})], \quad R_i(\eta_{i,t})=R_i^T(\eta_{i,t})>0$.
\item The pair $[A_i(\eta_{i,t}),~B_i(\eta_{i,t})]$ is  stabilizable.
\end{enumerate}
\begin{remark}
Assumption 1 is used to simplify the
derivation of the main result without loss of generality. Under  Assumption 2,  the stabilizability of each subsystem in every operation modes does not imply
that the Markovian jump large-scale system is stabilizable. The reason is twofold. Firstly, the
stability of Markovian jump systems depends not only on the matrix A but also on the mode
transition rate matrix. Secondly, the stability of all subsystems cannot guarantee the stability
of a large scale system due to the interconnections between subsystems.
\end{remark}

In this paper, the local uncertainties in each subsystem and the interconnections between subsystems
 are described by operators
$$\xi_i(t)=\phi_i^{\xi}(t,\zeta_i(\cdot)|_0^t,\eta_i(\cdot)|_0^t),$$
$$r_i(t)=\phi_i^{r}(t,\underline{\zeta}_i(\cdot)|_0^t,\eta_i(\cdot)|_0^t),$$
where $\underline{\zeta}_{i}(\cdot)\triangleq
\left[\zeta_{1}(\cdot),~\cdots~,\zeta_{i-1}(\cdot),~\zeta_{i+1}(\cdot),~\cdots,~\zeta_N(\cdot)\right]$. They
 are assumed to satisfy certain integral quadratic constraints, and the definitions are presented as below.
\begin{definition}[\cite{Xiong2009}] Given a set of matrices $\bar{S}_{i}\in\mathbb{S}^{+}$, $i\in\mathcal{N}$. A
  collection of local uncertainty inputs $\xi_{i}(t)$, $i\in\mathcal{N}$, is an
  admissible local uncertainty for the large-scale system if there exists a
  sequence $\set{t_{l}}_{l=1}^{\infty}$ such that $t_{l}\ge0$, $t_{l}\to\infty$,
  and for all $l$ and for all $i\in\mathcal{N}$,
\label{definition for local uncertainties}
\begin{equation}
\E\left(\int_0^{t_l}\left[\|\zeta_i(t)\|^2-\|\xi_i(t)\|^2\right]dt|x_0,\eta_0\right)\geq -x_{i0}^{T}\bar{S_i}x_{i0},
\end{equation}
  where $x_{0}=[x_{10}^{T},\ldots,x_{N0}^{T}]^{T}$, and $\eta_{0}=\eta(0)$. The set of the admissible local uncertainties is denoted by $\Xi$.
\end{definition}

\begin{definition}
[\cite{Xiong2009}] Given a set of matrices $\tilde{S}_{i}\in\mathbb{S}^{+}$, $i\in\mathcal{N}$.
  The large-scale system is said to have admissible
  interconnections between subsystems if there exists a sequence
  $\set{t_{l}}_{l=1}^{\infty}$ such that $t_{l}\ge0$, $t_{l}\to\infty$, and
  for all $l$ and for all $i\in\mathcal{N}$,
\label{definition for interconnections}
\begin{align}
&\E\left(\int_0^{t_l}\left[\left(\sum_{j=1,j\neq i}^{N}\|\zeta_j(t)\|^2\right)-\|r_i(t)\|^2\right]dt|x_0,\eta_0\right)  \geq -x_{i0}^{T}\tilde{S_i}x_{i0},
\end{align}
The set of the admissible interconnections is denoted by $\Pi$.
\end{definition}

Suppose the system $\mathcal{S}$ is subject to the uncertainty constraints (1) and (2), the aim of this paper is to design local mode dependent decentralized controllers
\begin{equation}
u_i(t)=K_i(\eta_{i,t})x_i(t),  \quad\quad i \in \mathcal{N},
\end{equation}
such that the system $\mathcal{S}$ is mean square stable and  the
following $H_\infty$ performance is satisfied: Given $\gamma>0$, if $x_0=0$,   for all
$i=1,\ldots N$,  then
\begin{equation}
\sup_{\substack{w(\cdot) \in L_2[0,\infty),\xi(\cdot)\in \Xi  \\  \|w(\cdot)\|\neq 0, r(\cdot) \in \Pi}} \frac{\left\{\E\left[ \int_0^{\infty}(\sum_{i=1}^{N}\|z_i(t)\|^2) dt\right]\right\}^{1/2}}{\left\{\int_0^{\infty}(\sum_{i=1}^{N}\|w_i(t)\|^2) dt\right\}^{1/2}}<\gamma .
\end{equation}

\section{The Main Results}
In this section,  a sufficient condition will be developed to design  local mode dependent controllers. Due to the local mode dependent controllers can not be obtained directly and the global operation mode process $\eta_t$ is Markovian, let us enlarge the mode state space of the subsystem in $\mathcal{S}$ and   consider an auxiliary class of uncertain systems $\tilde{\mathcal{S}}$  that contains the uncertain system $\mathcal{S}$ as a special case. There exist  necessary and sufficient conditions for the design of $H_\infty$ controllers for this auxiliary  system, see~\cite{Zhang2008,Marcos2008}.
The $i$-th subsystem of $\tilde{\mathcal{S}}$, $i \in \mathcal{N} $, is described by
\begin{equation*}
\label{auxiliary global mode system}
\tilde{\mathcal{S}}_i:\left\{
\begin{array}{l} \dot{\tilde{x}}_i(t)=\tilde{A}_i(\eta_t)\tilde{x}_i(t)+\tilde{B}_i(\eta_t)\left[\tilde{u}_i(t)+\tilde{\xi}^u_i(t)\right]+\tilde{E}_i(\eta_t)\tilde{r}_i(t)
  +\tilde{F}_i(\eta_t)\tilde{\xi}_i(t)+\tilde{G_i}(\eta_t)\tilde{w}_i(t),\\
 \tilde{z}_i(t)=\tilde{C}_i(\eta_t)\tilde{x}_i(t)+\tilde{D}_i(\eta_t)\left[\tilde{u}_i(t)+\tilde{\xi}^u_i(t)\right],\\
 \tilde{\zeta}_i(t)=\tilde{H}_i(\eta_t)\tilde{x}_i(t),
\end{array}
\right.
\end{equation*}
where $\tilde{A}_i(\mu)=A_i(\mu_i)$,  $\tilde{B}_i(\mu)=B_i(\mu_i)$, $\tilde{C}_i(\mu)=C_i(\mu_i)$, $\tilde{D}_i(\mu)=D_i(\mu_i)$, $\tilde{E}_i(\mu)=E_i(\mu_i)$, $\tilde{F}_i(\mu)=F_i(\mu_i)$, $\tilde{G}_i(\mu)=G_i(\mu_i)$, $\tilde{H}_i(\mu)=H_i(\mu_i)$, $\mu\in \mathcal{M}$ and $ \mu_i=\Psi_i^{-1}(\mu)\in\mathcal{M}_i$. The uncertainty input $\tilde{\xi}_i(t)$ and $\tilde{r}_i(t)$ are, respectively, generated by the same operator as $\xi_i(t)$ and $r_i(t)$, and satisfy the IQCs in (1), (2). $w_i(t)$, $z_i(t)$, $\zeta_i(t)$ are replaced with $\tilde{w}_i(t)$, $\tilde{z}_i(t)$, $\tilde{\zeta}_i(t)$, respectively.  $\tilde{\xi}_i^u$ is the uncertainty in the control input, satisfying the following IQCs.
\begin{definition} [\cite{Xiong2009}]
Given $\beta_i^u(\mu)\in \mathbb{R}^{+}$,$ i \in \mathcal{N}$,$ \mu \in \mathcal{M}$. A collection of uncertainty input $\tilde{\xi}_i^u(t)$ ,$ i \in \mathcal{N}$, is an admissible uncertainty input for the auxiliary large-scale system  $\tilde{\mathcal{S}}$ if
\label{definition for uncertainty in global system}
\begin{equation}
\E\left(\int_0^{t_l}\left(\beta_i^u(\eta_t)\|\tilde{x}_i(t)\|^2-\|\tilde{\xi}_i^u(t)\|^2\right)dt|\tilde{x}_0,\eta_0\right)\geq 0
\end{equation}
for all $ l$ and for all $ i \in \mathcal{N}$. The set of the admissible uncertainty inputs is denoted by $\tilde{\Xi}^{u}$.
\end{definition}

Without loss of generality, we assume that the same sequence
$\{t_{l}\}_{l=1}^{\infty}$ is used in Definitions 3, 4, 5.

Correspondingly, the global mode dependent decentralized controllers of the auxiliary large-scale system will be designed as the form below
\begin{equation}
\tilde{u}_i(t)=\tilde{K}_i(\eta_t)\tilde{x}_i(t), \quad\quad  i \in \mathcal{N}.
\end{equation}

The following lemma relates the stabilization of system  $\mathcal{S}$ with local mode dependent controllers (3) and the stabilization of the auxiliary system  $\tilde{\mathcal{S}}$ with global mode dependent controllers (6).

\begin{lemma} [\cite{Xiong2009}]
\label{lemma 1}
Suppose controllers (6) stabilize the uncertain large-scale system  $\tilde{\mathcal{S}}$ subject to the IQCs (1), (2), (5). If the control gains $K_i(\cdot)$ in (3) are chosen so that
\begin{equation}
\|\tilde{K}_i(\mu)-K_i(\mu_i)\|^2\leq\beta_i^u(\mu),
\end{equation}
where $\mu \in \mathcal{M}$, $ \mu_i=\Psi_i^{-1}(\mu)\in \mathcal{M}_i$, for all $i \in \mathcal{N}$, then the controllers in (3) stabilize the uncertain large-scale system  $\mathcal{S}$ subject to the IQCs (1) and (2).
\end{lemma}
The next lemma provides a way to construct local mode dependent controllers using the global ones.
\begin{lemma}[\cite{Xiong2009}]
\label{lemma 2}
Given  global mode dependent controllers (6), the corresponding local mode dependent controllers (3) can be constructed using the following transformation:
\begin{equation}
K_i(\nu_i)=\frac{\sum_{\mu=1}^{M}\left\{\tilde{K}_i(\mu)q_{\infty\mu}\mathbb{I}_i(\mu,\nu_i)\right\}}{\sum_{\mu=1}^M\left\{q_{\infty\mu}\mathbb{I}_i(\mu,\nu_i)\right\}},
\end{equation}
where $ \nu_i\in\mathcal{M}_i$, $i\in\mathcal{N}$, $\mathbb{I}_i(\mu,\nu_i)=1$ if $ \nu_i=\Psi_i^{-1}(\mu)$,and $\mathbb{I}_i(\mu,\nu_i)=0$ otherwise, $q_{\infty\mu}$ is the $\mu$ component of the vector $q_\infty$.  Then $K_i(\nu_i)=\lim_{t \to \infty} \E(\tilde{K}_i(\eta_t)|\eta_{i,t}=\nu_i)$.
\end{lemma}

 The sufficient condition to design the local mode dependent controllers (3) is given by the following theorem.

\begin{theorem}
\label{theorem 1}
Let $\gamma >0$ be a prescribed level of $H_{\infty}$-norm bound, suppose there exist matrices $ X_i(\mu)\in \mathbb{S}^{+} $, $ \bar{X}_i(\mu)\in \mathbb{S}^{+} $, scalars $\bar{\beta}_i(\mu) \in \mathbb{R}^{+}$, $\tilde{\beta}_i(\mu) \in \mathbb{R}^{+}$, $\bar{\tau}_i \in \mathbb{R}^{+}$,  $\bar{\theta}_i \in \mathbb{R}^{+}$, $\bar{\phi}_i \in \mathbb{R}^{+}$, $i\in \mathcal{N}$, $ \mu \in \mathcal{M}$, such that the matrix inequalities with rank constraints (9)-(13) are satisfied
\begin{equation}
\left[
\begin{array}{cccc}
\Upsilon_{i11}(\mu) & \Upsilon_{i12}(\mu) & \Upsilon_{i13}(\mu)\\
\star & \Upsilon_{i22}(\mu) & 0   \\
\star & \star    & \Upsilon_{i33}(\mu)
\end{array}
\right]<0,
\end{equation}
\begin{equation}
\bar{\phi}_i\tilde{R}_i(\mu)-I <0,
\end{equation}
\begin{equation}
\left[
\begin{array}{cc}
 \bar{\phi}_iI & \Delta_i(\mu) \\
\Delta_i^{T}(\mu) & \tilde{\beta}_i(\mu) I
\end{array}
\right] \ge 0,
\end{equation}
\begin{equation}
 \rank\left(
  \begin{bmatrix}
   \bar{X}_i(\mu)  & I\\
    I & X_i(\mu)
  \end{bmatrix}\right) \le n_i,
  \end{equation}
  \begin{equation}
   \rank\left(
 \begin{bmatrix}
\bar{\beta}_i(\mu)  & 1\\
1 & \tilde{\beta}_i(\mu)
\end{bmatrix}\right) \le 1,
 \end{equation}
where
\begin{align*}
& \tilde{R}_i(\mu)=\tilde{D}_i(\mu)^T\tilde{D}_i(\mu),\\
&\Upsilon_{i11}(\mu) =X_i(\mu)\tilde{A}_i^{T}(\mu)+\tilde{A_i}(\mu)X_i(\mu)+\bar{\theta}_i\tilde{E}_i(\mu)\tilde{E}_i^{T}(\mu)  +\bar{\tau}_i\tilde{F}_i(\mu)\tilde{F}_i^{T}(\mu)+\gamma^{-2}\tilde{G}_i(\mu)\tilde{G}_i^{T}(\mu)
\\&\quad \quad \quad \quad -\tilde{B}_i(\mu)\tilde{R}_i^{-1}(\mu)\tilde{B}_i^T(\mu)+q_{\mu \mu}X_i(\mu), \\
&\Upsilon_{i12}(\mu) = X_i(\mu)
                     \left[
                     \tilde{C}_i^{T}(\mu) \quad I \quad \overbrace{\tilde{H}_i^{T}(\mu) \quad \ldots \quad \tilde{H}_i^{T}(\mu)}^N
                     \right], \\
&\Upsilon_{i22}(\mu)= -\diag[I ~~ \bar{\beta}_i I ~~ \bar{\tau}_i I ~~ \bar{\theta}_1 I ~~ \ldots ~~ \bar{\theta}_{i-1}I~~ \bar{\theta}_{i+1}I ~~ \ldots ~~ \bar{\theta}_{N}I],\\
&\Upsilon_{i13}(\mu)=X_i(\mu)\left[\sqrt{q_{\mu,1}}I ~~  \ldots ~~ \sqrt{q_{\mu,\mu-1}}I ~~ \sqrt{q_{\mu,\mu+1}}I ~~ \ldots ~~ \sqrt{q_{\mu,M}}I \right], \\
&\Upsilon_{i33}(\mu)= -\diag[X_i(1)~~ \ldots ~~ X_i(\mu-1)~~ X_i(\mu+1)~~\ldots ~~ X_i(M) ],\\
&\Delta_i(\mu)     =\frac{1}{\sum_{\nu=1}^{M} \left\{ \mathbb{I}_{i}(\nu,\mu_{i})q_{\infty\nu}
      \right\}}\sum_{\nu=1}^{M}\left\{\mathbb{I}_i(\nu,\mu_i)q   _{\infty\nu}
\left[\tilde{R}_i^{-1}(\nu)\tilde{B}_i^T(\nu)\bar{X}_i(\nu) -\tilde{R}_i^{-1}(\mu)\tilde{B}_i^T(\mu)\bar{X}_i(\mu)
 \right] \right\}.
\end{align*}

Then the global mode dependent control gains $\tilde{K}_i(\mu)$ are given by
\begin{align}
\tilde{K}_i(\mu)=-\tilde{R}_i(\mu)^{-1}\tilde{B}_i^T(\mu)\bar{X}_i(\mu).
\end{align}

Moreover,  the local mode dependent controllers (3) can be constructed using (8) so that the uncertain system large-scale system $\mathcal{S}$ subject to the IQCs (1)~(2) is mean square stable and satisfies the $H_\infty$-performance (4).
\end{theorem}

 \begin{proof}
\label{proof of theorem 1}
The proof is divided into two steps. In step 1, we show that the global mode dependent controllers defined in (6) with the form of (14) stabilize the auxiliary system  $\tilde{\mathcal{S}}$ subject to the IQCs (1), (2), (5) and satisfy $H_{\infty}$ performance (4). In step 2, we  show that  the local mode dependent controllers  (3) constructed by (8)  stabilize the system $\mathcal{S}$ subject to the IQCs (1)~(2) and satisfy $H_{\infty}$ performance (4).

Step 1

Using Schur complement equivalence, it follows from inequality (9), $$\Upsilon_{i11}-\Upsilon_{i12}\Upsilon_{i22}^{-1}\Upsilon_{i12}^{T}
-\Upsilon_{i13}\Upsilon_{i33}^{-1}\Upsilon_{i13}^{T}<0 . $$
 That is,
\begin{IEEEeqnarray}{rl}
& \quad X_i(\mu)\tilde{A}_i^T(\mu)+\tilde{A}_i(\mu) X_i(\mu)+\sum_{\nu=1}^M q_{\mu\nu}X_i(\mu)X_i(\nu)^{-1}X_i(\mu)+\bar{\theta}_i\tilde{E}_i(\mu)\tilde{E}_i^T(\mu)
\nonumber\\&+\:
X_i(\mu)\left[\tilde{C}_i^T(\mu)\tilde{C}_i(\mu)+\bar{\beta}_i^{-1}(\mu)I+(\bar{\tau}_i^{-1}+\sum_{j=1,j\neq i}^N \bar{\theta}_j^{-1}) \tilde{H}_i^T(\mu)\tilde{H}_i(\mu) \right]X_i(\mu)
\nonumber\\&+\:
\bar{\tau}_i\tilde{F}_i(\mu)\tilde{F}_i^T(\mu)+ \gamma^{-2}\tilde{G}_i(\mu)\tilde{G}_i^T(\mu)-\tilde{B}_i(\mu)\tilde{R}_i^{-1}(\mu)\tilde{B}_i^T(\mu)
< 0.
\end{IEEEeqnarray}

The rank constraints (12) and (13), together with $X_i(\mu)>0$, $\bar{X}_i(\mu)>0$  and $\bar{\beta}_i(\mu)>0$, $\tilde{\beta}_i(\mu)>0$, are  equivalent to
 $$ \bar{X}_i(\mu)=(X_i(\mu))^{-1}, \quad \bar{\beta}_i(\mu)=(\tilde{\beta}_i(\mu))^{-1}.$$

Let $\tau_i=\bar{\tau}_i^{-1}$, $\theta_i=\bar{\theta}_i^{-1}$, $\phi_i=\bar{\phi}_i^{-1}$, $\beta_i^u(\mu)=\bar \phi_i \tilde{\beta}_i(\mu)$ and $\hat{\theta}_i=\sum_{j=1,j\neq i}^N  \theta_j$. Pre-multiplying and post-multiplying (15) by $\bar{X}_i(\mu)$, with $\tilde{K}_i(\mu)$  defined as the form of (14),  we can get:

\begin{IEEEeqnarray}{rl}
&\quad\left[\tilde{A}_i(\mu)+\tilde{B}_i(\mu)\tilde{K}_i(\mu)\right]^T\bar{X}_i(\mu)+\bar{X}_i(\mu)\left[\tilde{A}_i(\mu)+\tilde{B}_i(\mu)\tilde{K}_i(\mu)\right]
\nonumber\\&+\:\tilde{C}_i^T(\mu)\tilde{C}_i(\mu)+\phi_i\beta_i^u(\mu)I+\sum_{\nu=1}^Mq_{\mu\nu}\bar{X}_i(\nu)+(\tau_i+\hat{\theta}_i)\tilde{H}_i^T(\mu)\tilde{H}_i(\mu) \nonumber\\&+\:\bar{X}_i(\mu)\left[\theta_i^{-1}\tilde{E}_i(\mu)\tilde{E}_i^T(\mu)+\tau_i^{-1}\tilde{F}_i(\mu)\tilde{F}_i^T(\mu)+\gamma^{-2}\tilde{G}_i(\mu)\tilde{G}_i^T(\mu)\right]\bar{X}_i(\mu)
\nonumber\\&+\:\tilde{K}_i^T(\mu)\tilde{R}_i(\mu)\tilde{K}_i(\mu)-\left[\bar{X}_i(\mu)\tilde{B}_i(\mu)
+\tilde{K}_i^T(\mu)\tilde{R}_i(\mu)\right]\left[-\phi_i I+\tilde{R}_i(\mu)\right]^{-1}
\nonumber\\& \times\: \left[\bar{X}_i(\mu)\tilde{B}_i(\mu)+\tilde{K}_i^T(\mu)\tilde{R}_i(\mu)\right]^T
<0.
\end{IEEEeqnarray}

In the derivation of (16), we have used the following equations:
\begin{IEEEeqnarray}{rl}
&\quad \left[\tilde{B}_i(\mu)\tilde{K}_i(\mu)\right]^T\bar{X}_i(\mu)+\bar{X}_i(\mu)\left[\tilde{B}_i(\mu)\tilde{K}_i(\mu)\right]+\tilde{K}_i^T(\mu)\tilde{R}_i(\mu)\tilde{K}_i(\mu) \nonumber\\&=\:-\bar{X}_i(\mu)\tilde{B}_i(\mu)\tilde{R}_i^{-1}(\mu)\tilde{B}_i^T(\mu)\bar{X}_i(\mu),\nonumber
\end{IEEEeqnarray}
and
$$\bar{X}_i(\mu)\tilde{B}_i(\mu)+\tilde{K}_i^T(\mu)\tilde{R}_i(\mu)=0.$$

With the inequality (10)  and according to Schur complement equivalence, the inequality (16) can be transformed to
\begin{equation}
\left[
\begin{array}{cccc}
\hat{\Upsilon}_{i11}(\mu) & \hat{\Upsilon}_{i12}(\mu) & \hat{\Upsilon}_{i13}(\mu)\\
\star & \hat{\Upsilon}_{i22}(\mu) & 0   \\
\star & \star & -I
\end{array}
\right]<0,
\end{equation}
where
\begin{align*}
\hat{\Upsilon}_{i11}(\mu)&=\left[\tilde{A}_i(\mu)+\tilde{B}_i(\mu)\tilde{K}_i(\mu)\right]^T\bar{X}_i(\mu)+\bar{X}_i(\mu)
   \left[\tilde{A}_i(\mu)+\tilde{B}_i(\mu)\tilde{K}_i(\mu)\right]+\sum_{\nu=1}^Mq_{\mu\nu}\bar{X}_i(\nu)
   \\&+\tilde{C}_i^T(\mu)\tilde{C}_i(\mu)
   +\phi_i\beta_i^u(\mu)I+\tilde{K}_i^T(\mu)\tilde{R}_i(\mu)\tilde{K}_i(\mu)+(\tau_i+\hat{\theta}_i)\tilde{H}^T_i(\mu)\tilde{H}_i(\mu),\\
\hat{\Upsilon}_{i12}(\mu)&= \phi_i^{-1/2}\bar{X}_i(\mu)\tilde{B}_i+\phi_i^{-1/2}\tilde{K}_i^T(\mu)\tilde{R}_i(\mu),\\
\hat{\Upsilon}_{i22}(\mu)&= - I+\phi_i^{-1}\tilde{R}_i(\mu),\\
\hat{\Upsilon}_{i13}(\mu)&=\bar{X}_i(\mu)\begin{bmatrix}
\theta_i^{-1/2}\tilde{E}_i(\mu) & \tau_i^{-1/2}\tilde{F}_i(\mu) & \gamma^{-1}\tilde{G}_i(\mu)
\end{bmatrix}.
\end{align*}

Using Schur complement equivalence again and Assumption 1, we get
\begin{equation}
\left[
\begin{array}{cccc}
\bar{\Upsilon}_{i11}(\mu) & \bar{\Upsilon}_{i12}(\mu) & \bar{\Upsilon}_{i13}(\mu)\\
\star & \bar{\Upsilon}_{i22}(\mu) & \bar{\Upsilon}_{i23}(\mu)   \\
\star & \star & -I
\end{array}
\right]<0,
\end{equation}
where
\begin{align*}
\bar{\Upsilon}_{i11}(\mu)&=\left[\tilde{A}_i(\mu)+\tilde{B}_i(\mu)\tilde{K}_i(\mu)\right]^T\bar{X}_i(\mu)+\bar{X}_i(\mu)\left[\tilde{A}_i(\mu)+\tilde{B}_i(\mu)\tilde{K}_i(\mu)\right]+\sum_{\nu=1}^Mq_{\mu\nu}\bar{X}_i(\nu),\\
\bar{\Upsilon}_{i12}(\mu)&= \bar{X}_i(\mu)\left[\phi_i^{-1/2}\tilde{B}_i(\mu) \quad\theta_i^{-1/2}\tilde{E}_i(\mu) \quad \tau_i^{-1/2}\tilde{F}_i(\mu) \quad  \gamma^{-1}\tilde{G}_i(\mu)\right],\\
\bar{\Upsilon}_{i22}(\mu)&= - I,\\
\bar{\Upsilon}_{i13}(\mu)&=\begin{bmatrix}
\tilde{C}_i(\mu)+\tilde{D}_i(\mu)\tilde{K}_i(\mu) \\ (\tau_i+\hat{\theta}_i)^{1/2}\tilde{H}_i(\mu) \\ [\phi_i\beta_i^u(\mu)]^{1/2}I
\end{bmatrix}^T,\\
\bar{\Upsilon}_{i23}(\mu)&=\begin{bmatrix}
\phi_i^{-1/2}\tilde{D}_i^T(\mu) & 0& 0\\
0 & 0 &  0 \\
0  & 0  & 0  \\
0  & 0  & 0
\end{bmatrix}.
\end{align*}

Let us define
  \begin{align*}
&\bar{w}_i(t)\triangleq\left[\phi_i^{1/2}\tilde{\xi}_i^u(t)^T\quad\theta_i^{1/2}\tilde{r}_i(t)^T\quad\tau_i^{1/2}\tilde{\xi}_i(t)^T\quad\gamma\tilde{w}_i(t)^T\right]^T,\\
&\bar{z_i}(t)\triangleq\left[\tilde{z}_i(t)^T\quad(\tau_i+\hat{\theta}_i)^{1/2}\tilde{\zeta}_i(t)^T\quad[\phi_i\beta_i^u(\eta_t)]^{1/2}\tilde{x}_i(t)^T\right]^T,
\end{align*}
and  define $\bar{w}=[\bar{w}_1^T,\ldots ,\bar{w}_N^T]^T$ as the disturbance input, and $\bar{z}=[\bar{z}_1^T,\ldots ,\bar{z}_N^T]^T$ as the output, so the auxiliary system $\tilde{\mathcal{S}}$ with augmented outputs can be rewritten as the following system  $\bar{\mathcal{S}}$
\begin{equation*}
\bar{\mathcal{S}}:\left\{
\begin{array}{l}
\dot{\tilde{x}}(t)=A(\eta_t)\tilde{x}(t)+B(\eta_t)\tilde{u}(t)+B_2(\eta_t)\bar{w}(t),\\
\bar{z}(t)=C(\eta_t)\tilde{x}(t)+D(\eta_t)\tilde{u}(t)+D_2(\eta_t)\bar{w}(t),
\end{array}
\right.
\end{equation*}
where
\begin{align*}
A(\mu)&=\diag\left[\tilde{A}_1(\mu), \ldots ,\tilde{A}_N(\mu)\right] ,
 & B(\mu)&=\diag\left[\tilde{B}_1(\mu),\ldots,\tilde{B}_{N}(\mu)\right],
& B_2(\mu)&=\diag\left[\bar{B}_{2,1}(\mu),\ldots ,\bar{B}_{2,N}(\mu)\right],\\
C(\mu)&=\diag\left[\bar{C}_1(\mu), \ldots , \bar{C}_N(\mu)\right],
&D(\mu)&=\diag\left[\bar{D}_1(\mu),\ldots,\bar{D}_N(\mu)\right],
&D_2(\mu)&=\diag\left[\bar{D}_{2,1}(\mu),\ldots,\bar{D}_{2,N}(\mu)\right],\\
\tilde{x}&=\left[\tilde{x}_1^T,\ldots ,\tilde{x}_N^T\right]^T,\quad &\tilde{u}&=\left[\tilde{u}_1^T,\ldots ,\tilde{u}_N^T\right]^T,
\end{align*}
and the matrix coefficients are:
\begin{align*}
& \bar{B}_{2,i}(\mu)=\begin{bmatrix}
\phi_{i}^{-1/2}\tilde{B}_i(\mu)
& \theta_i^{-1/2}\tilde{E}_i(\mu)
 &\tau_i^{-1/2}\tilde{F}_i(\mu)
 &\gamma^{-1}\tilde{G}_i(\mu)
 \end{bmatrix},
 \\
&\bar{C}_{i}(\mu)=
\begin{bmatrix}
\tilde{C}_{i}(\mu)\\
(\tau_i+\hat{\theta}_i)^{1/2}\tilde{H}_i(\mu)\\
 [\phi_i\beta_i^{u}(\mu)]^{1/2}I
\end{bmatrix},
\\
&\bar{D}_{i}(\mu)=
\begin{bmatrix}
\tilde{D}_i(\mu)\\
 0\\
 0
\end{bmatrix},\\
&\bar{D}_{2,i}(\mu)=
\begin{bmatrix}
\phi_i^{-1/2}\tilde{D}_i(\mu) &0 &0 &0 \\
 0 & 0 & 0& 0\\
 0 &0 &0 & 0
\end{bmatrix}.
 \end{align*}

 Note that the auxiliary system  $\tilde{\mathcal{S}}$ and the system $\bar{\mathcal{S}}$ have the same stability property.

 Then the inequality (18) can be rewritten as
 \begin{align*}
\begin{bmatrix}
   \bar{\Upsilon}_{i11}(\mu) & \bar{X}_i(\mu)\bar B_{2i}(\mu) & [\bar{C}_i(\mu)+\bar{D}_i(\mu)\tilde{K}_i(\mu)]^T\\
   \star & -I & \bar D_{2i}^T(\mu)\\
   \star   & \star & -I
  \end{bmatrix} <0,\quad i=1,2,\ldots,N.
\end{align*}

 With the positive definite matrices $\bar{X}(\mu)$ defined as $\diag[\bar{X}_1(\mu),\ldots,\bar{X}_N(\mu)]$ and the controller  $\tilde{K}(\mu)$ defined as $\diag[\tilde{K}_1(\mu),\ldots,\tilde{K}_N(\mu)]$, we have
\begin{align*}
\begin{bmatrix}
    \bar{\Upsilon}_{11}(\mu) & \bar{X}(\mu)B_2(\mu) & [C(\mu)+D(\mu)\tilde{K}(\mu)]^T\\
    \star & -I & D_2^T(\mu)\\
   \star & \star & -I
  \end{bmatrix} <0,
\end{align*}
where $\bar{\Upsilon}_{11}(\mu)=\diag[\bar{\Upsilon}_{111}(\mu),\ldots,\bar{\Upsilon}_{N11}(\mu)]$.

According to  Lemma 1, the system $\bar{\mathcal{S}}$ is mean square stable and $\|\bar{\mathcal{S}}\|_{\infty}<1$. This implies
\begin{align*}
&\sup_{\substack{  \bar{w}(\cdot)\in L_2[0,\infty)\\\|\bar{w}(\cdot)\|\neq 0,~ x_0=0}}\frac{\left\{\E(\int_0^{\infty}\|\bar{z}(t)\|^2)dt\right\}^{1/2}}{\left\{\int_0^{\infty}\|\bar{w}(t)\|^2dt\right\}^{1/2}} =
\sup_{\substack{  \bar{w}(\cdot)\in L_2[0,\infty)\\\|\bar{w}(\cdot)\|\neq 0,~ x_0=0}}\frac{\left\{\E\left[\int_0^{\infty}(\sum_{i=1}^N\|\bar{z}_i(t)\|^2)dt\right]\right\}^{1/2}}{\left\{\int_0^{\infty}(\sum_{i=1}^N\|\bar{w}_i(t)\|^2)dt\right\}^{1/2}}
<1.
\end{align*}

Notice  the definition of  $\bar{z}_i$ and  $\bar{w}_i$  , we  get
\begin{multline}
\E\left\{\int_0^{\infty}\sum_{i=1}^N \left[ \|\tilde z_i(t)\|^2+(\tau_i +\hat{\theta}_i)\|\tilde\zeta_i(t)\|^2 +\phi_i \beta_i^u(\eta_t)\|\tilde{x}_i(t)\|^2 \right] dt\right\}\\
\quad\quad -\E\left\{\int_0^{\infty}\sum_{i=1}^N\left[\gamma^2\| \tilde w_i(t)\|^2+\tau_i\|\tilde\xi_i(t)\|^2 +\theta_i\|\tilde r_i(t)\|^2+\phi_i\|\tilde \xi_i^u(t)\|^2\right]dt \right\}\nonumber\\
=\E\left\{\int_0^{\infty}\sum_{i=1}^N\left[(\|\tilde z_i(t)\|^2-\gamma^2\| \tilde w_i(t)\|^2)+\tau_i(\|\tilde\zeta_i(t)\|^2-\|\tilde\xi_i(t)\|^2)\right.\right.
 \\ \quad\quad \left.+\theta_i(\sum^{N}_{j=1,j\neq i}\left. \|\tilde\zeta_j(t)\|^2-\|\tilde r_i(t)\|^2)+\phi_i(\beta_i^u(\eta_t)\|\tilde{x}_i(t)\|^2- \|\tilde \xi_i^u(t)\|^2)\right]dt\right\}
\nonumber< 0.
\end{multline}

According to (1), (2) and (5)  we  have
\begin{align}
\E\left\{\int_0^{\infty}\sum_{i=1}^N \left[ \|\tilde z_i(t)\|^2-\gamma^2\| \tilde w_i(t)\|^2 \right] dt\right\}
<
\sum_{i=1}^N\left(\tau_i x_{i0}^T\bar{S}_i x_{i0}+\theta_i x_{i0}^T\tilde{S}_i x_{i0}\right).
\end{align}

When $x_0=0$, the right part of the inequality (19) is equal to zero, it means:
\begin{equation*}
\sup_{\substack{\tilde{w}(\cdot) \in L_2[0,\infty), \tilde{\xi}(\cdot)\in \Xi  \\ \|\tilde{w}(\cdot)\| \neq 0 , \tilde{r}(\cdot) \in \Pi, \tilde{\xi}^u(\cdot) \in \tilde{\Xi}^u }} \frac{\left\{\E \left[\int_0^{\infty}(\sum_{i=1}^{N}\|\tilde z_i(t)\|^2 )dt \right]\right\}^{1/2}}{\left\{\int_0^{\infty}(\sum_{i=1}^{N}\|\tilde w_i(t)\|^2) dt\right\}^{1/2}}<\gamma.
\end{equation*}

So the controllers (6) with the gains given by (14) stabilize system  $\tilde{\mathcal{S}}$  and satisfy $H_{\infty}$-performance (4).

Step 2

\label{step 2}
In the LMI (11), it can be obtained that
\begin{align*}
\Delta_i(\mu)&=\frac
{\sum_{\nu=1}^{M}\left\{\mathbb{I}_i(\nu,\mu_i)q_{\infty\nu}
   [\tilde{K}_i(\mu)-\tilde{K}_i(\nu)]
\right\}}{\sum_{\nu=1}^{M} \left\{ \mathbb{I}_{i}(\nu,\mu_{i})q_{\infty\nu}
      \right\}}=\tilde{K}_i(\mu)-K_i(\mu_i),
\end{align*}
and
\begin{equation*}
\|\Delta_i(\mu)\|^2=\|\tilde{K}_i(\mu)-K_i(\mu_i)\|^2\leq \bar{\phi_i}\tilde{\beta_i}(\mu)=\beta_i^u(\mu).
\end{equation*}

According to Lemma 3, the constructed controller (3) stabilizes the uncertain system $\mathcal{S}$ subject to the IQCs (1) and (2).

As for the $H_\infty$ performance of system $\mathcal{S}$,   a particular uncertainty input of the form $\tilde{\xi}_i^u(t)$ is chosen to be $$\tilde{\xi}_i^u(t)=-\Delta_i(\eta_t)\tilde{x_i}(t).$$

Then we get
$$\|\tilde{\xi}_i^u(t)\|^2=\|\Delta_i(\eta_t)\tilde{x_i}(t)\|^2\leq \beta_i^u(\eta_t)\|\tilde{x_i}(t)\|^2,$$
 which means $\tilde{\xi}_i^u(t)$ is admissible uncertainty input for $\tilde{\mathcal{S}}$ according to Definition 5, and we have
$$ \tilde{u}_i(t)+\tilde{\zeta}_i^u(t)=\left[\tilde{K}_i(\eta_t)-\Delta_i(\eta_t)\right]\tilde{x_i}(t)=K_i(\eta_{i,t})\tilde{x}_i(t).$$

 If the same disturbances and sample path are applied to system $\mathcal{S}$ and system  $\tilde{\mathcal{S}}$ we can get $\tilde{x}_i(t)=x_i(t)$.

The controlled output of the closed-loop auxiliary large-scale system  $\tilde{\mathcal{S}}$ is
$$\tilde{z}_i(t)=[\tilde{C}_i(\eta_t)+\tilde{D}_i(\eta_t)K_i(\eta_{i,t})]\tilde{x}_i(t),$$
and the controlled output of the closed-loop system  $\mathcal{S}$ is
$$z_i(t)=[C_i(\eta_{i,t})+D_i(\eta_{i,t})K_i(\eta_{i,t})]x_i(t),$$
with $\tilde{C}_i(\eta_t)=C_i(\eta_{i,t})$, $\tilde{D}_i(\eta_t)=D_i(\eta_{i,t})$, so we  get $\|\tilde{z}_i(t)\|^2=\|z_i(t)\|^2$.
Thus  for the system  $\mathcal{S}$, it is satisfied that

\begin{equation*}
\sup_{\substack{w(\cdot) \in L_2[0,\infty),\xi(\cdot)\in \Xi  \\ \|w(\cdot)\| \neq 0 , r(\cdot) \in \Pi}} \frac{\left\{\E \left[ \int_0^{\infty}(\sum_{i=1}^{N}\|z_i(t)\|^2) dt\right]\right\}^{1/2}}{\left\{\int_0^{\infty}(\sum_{i=1}^{N}\|w_i(t)\|^2) dt\right\}^{1/2}}<\gamma.
\end{equation*}

It means the local mode dependent controllers (3)  stabilize the uncertain large-scale system $\mathcal{S}$ subject to the IQCs (1)  and (2) with $H_{\infty}$-performance (4).

\end{proof}
\begin{remark}
Due to the existence of the rank constraints (12) and (13), the solution set to (9)-(13) is non-convex, bringing difficulties to find numerical solutions. Actually,
LMIs with rank constraints have appeared widely in control area~\cite{Xiong2009,li20072,Xiong2010}, and there have been a lot of efforts to deal with this kind of problems, see ~\cite{Iwasaki1995,Ghaoui1997,Orsi2006} and the reference therein. We can use the rank constrained LMI solver LMIRank~\cite{Orsi} to solve the conditions in Theorem 1.
\end{remark}
\section{Numerical example}
We consider a large-scale system, which has three subsystems and each subsystem has two modes. The interconnection of this large-scale system can be seen in Fig. 2 (the structure of each subsystem can be seen in Fig. 1).
\begin{center}
\begin{picture}(100,155)
\put(-200,-330){\resizebox{150mm}{200mm}{\includegraphics[width=0.5\textwidth]{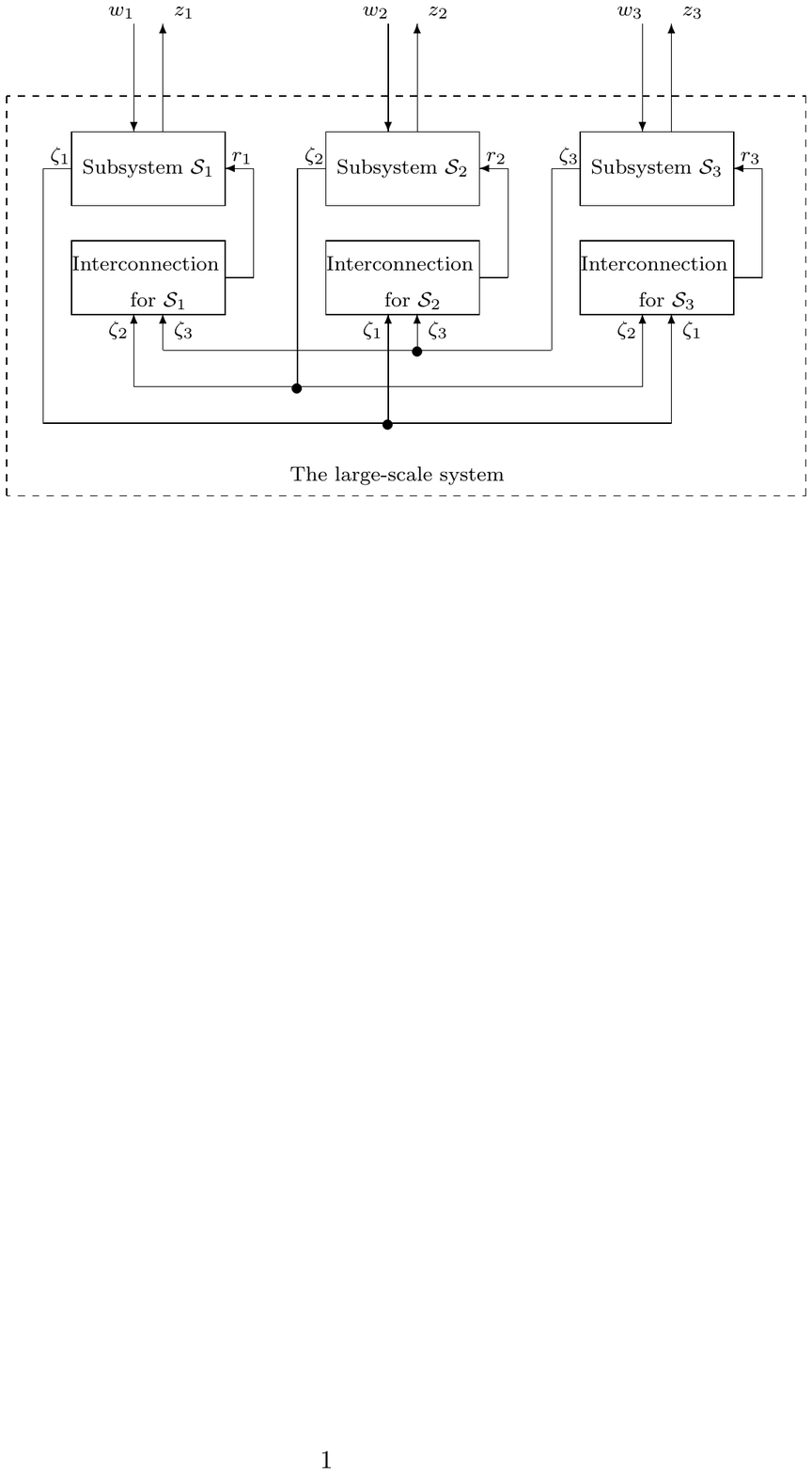}}}
\put(-50,-10){\shortstack{ Fig. 2. The structure of the system for example. }}
\end{picture}
\end{center}

The initial condition of the system is assumed to be $x_{10}=[5,-5]^T$, $x_{20}=[3,-3]^T$, $x_{30}=[1,-1]^T$, and  $\eta_{10}=\eta_{20}=\eta_{30}=1$. The  state transition rate matrix of the global modes is given by:
\begin{align*}
Q=
\begin{bmatrix}
-0.35 & 0.2 &0.1 & 0.05\\
0.5 & -2.4 & 0.7 & 1.2\\
0.4 & 0.3 & -1.45 & 0.75\\
0.1 & 0.2 & 0.3 & -0.6
\end{bmatrix}.
\end{align*}

We assume there are constraints on the operation modes of the subsystems: the operation mode of subsystem $\mathcal{S}_1 $ is assumed to depend on the operation modes of subsystems $\mathcal{S}_2 $  and  $\mathcal{S}_3 $. Explicitly, $\eta_1(t)=1$ if $\eta_2(t)=\eta_3(t)$, and $\eta_1(t)=2$ otherwise. Therefore, the operation mode pattern set $\mathcal{M}_p$ is given by $\set{(1,1,1),(1,2,2),(2,1,2),(2,2,1)}$, so  $\mathcal{M}=\set{1,2,3,4}$.
The relationship between the local operation modes $(\nu_1,\nu_2,\nu_3)$ and the global operation modes $\mu$ can be described by Table 1.
\begin{table}[h]
\caption{Relationship between local operation modes and global ones}
\centering
\begin{tabular}{c|c}
\hline
\hline
$(\nu_1,\nu_2,\nu_3)$ & $\mu$\\
\hline
(1,1,1) &1\\
\hline
(1,2,2) &2\\
\hline
(2,1,2) &3\\
\hline
(2,2,1) &4\\
\hline
\hline
\end{tabular}
\end{table}

The system data  are chosen as follows:
\begin{align*}
 A_1(1)&=
\begin{bmatrix}
1 & 0\\
1 & 1.4
\end{bmatrix},
 &A_1(2)&=
\begin{bmatrix}
1 & 0\\
1 & -2.4
\end{bmatrix},
& A_2(1)&=
\begin{bmatrix}
1 & 0\\
1 & 1.2
\end{bmatrix},
  &A_2(2)&=
\begin{bmatrix}
1 & 0\\
1 & -3.2
\end{bmatrix},\\
 A_3(1)&=
\begin{bmatrix}
1 & 0\\
1 & 1.2
\end{bmatrix},
& A_3(2)&=\begin{bmatrix}
1 & 0\\
1 & -1.2
\end{bmatrix},
& B_1(1)&=\begin{bmatrix}
 1.9\\
 -1.4
 \end{bmatrix},
  & B_1(2)&=\begin{bmatrix}
  1.6\\
  -1.8\end{bmatrix},\\
B_2(1)&=
 \begin{bmatrix}
 1.8\\
 1.7\end{bmatrix},
  &B_2(2)&=
  \begin{bmatrix}
  1.5\\
  -1.7
  \end{bmatrix},
& B_3(1)&=
 \begin{bmatrix}
 1.5\\
 1.7
 \end{bmatrix},
    &  B_3(2)&=
    \begin{bmatrix}
    1.4\\
    1.7
    \end{bmatrix},\\
 C_1(1)&=
 \begin{bmatrix}
 0.6 & 0.4\\
 0& 0
 \end{bmatrix},
 & C_1(2)&=
 \begin{bmatrix}
 0.6 & -0.4\\
 0 & 0
 \end{bmatrix},
& C_2(1)&=
 \begin{bmatrix}
 -0.7 & 0.5\\
 0& 0
 \end{bmatrix},
   &C_2(2)&=
  \begin{bmatrix}
  0.6& 0.4\\
  0&    0
  \end{bmatrix},\\
 C_3(1)&=
 \begin{bmatrix}-0.7 &  0.5\\
 0& 0
 \end{bmatrix},
  & C_3(2)&=
  \begin{bmatrix}
  0.6& 0.4\\
  0&    0
  \end{bmatrix},
 &D_1(1)&=  \begin{bmatrix} 0\\-0.3 \end{bmatrix},
   & D_1(2)&= \begin{bmatrix}0\\0.3\end{bmatrix},\\
  D_2(1)&=  \begin{bmatrix}0\\0.6\end{bmatrix},
     &D_2(2)&=  \begin{bmatrix}0\\-0.6\end{bmatrix},
  &D_3(1)&=  \begin{bmatrix}0\\0.6\end{bmatrix},
   & D_3(2)&=  \begin{bmatrix}0\\-0.6\end{bmatrix},\\
 E_1(1)&=  \begin{bmatrix}0.1;\\0.1\end{bmatrix},
  & E_1(2)&=  \begin{bmatrix}0.1\\-0.1\end{bmatrix},
 & E_2(1)&=  \begin{bmatrix}0.2\\0.1\end{bmatrix},
  & E_2(2)&=  \begin{bmatrix}0.1\\0.2\end{bmatrix},\\
  E_3(1)&=  \begin{bmatrix}0.1\\0.1\end{bmatrix},
   & E_3(2)&=  \begin{bmatrix}0.1\\0.2\end{bmatrix},
 &F_1(1)&=  \begin{bmatrix}0.5\\-0.1\end{bmatrix},
   & F_1(2)&=  \begin{bmatrix}0.5\\-0.1\end{bmatrix},\\
 F_2(1)&=  \begin{bmatrix}-0.3\\-0.2\end{bmatrix},
   &F_2(2)&=  \begin{bmatrix}0.2\\0.2\end{bmatrix},
  &F_3(1)&=  \begin{bmatrix}-0.1\\-0.2\end{bmatrix},
  & F_3(2)&=  \begin{bmatrix}0.1\\0.2\end{bmatrix},\\
G_1(1)&=  \begin{bmatrix}-0.3\\0.3\end{bmatrix},
  & G_1(2)&=  \begin{bmatrix}0.3\\-0.4\end{bmatrix},
 & G_2(1)&=  \begin{bmatrix}0.3\\-0.4\end{bmatrix},
   &G_2(2)&=  \begin{bmatrix}0.4\\0\end{bmatrix},\\
  G_3(1)&=  \begin{bmatrix}-0.3\\-0.3\end{bmatrix},
  & G_3(2)&=  \begin{bmatrix}0.4\\0\end{bmatrix},
   &H_1(1)&=\begin{bmatrix}0.8 &   0.3\end{bmatrix},
  & H_1(2)&=\begin{bmatrix}0.7 & 0.2\end{bmatrix},\\
  H_2(1)&=\begin{bmatrix}0.2 & 0.6\end{bmatrix},
  &H_2(2)&=\begin{bmatrix}0.8 &        0.2\end{bmatrix},
 &H_3(1)&=\begin{bmatrix}0.1 &   0.5\end{bmatrix},
 & H_3(2)&=\begin{bmatrix}0.7 &        0.3\end{bmatrix}.
\end{align*}

The local uncertainties and interconnections were chosen to be the following form,
\begin{equation}
\left\{
\begin{array}{ll}
\dot{x}_{\xi i}(t)=-10 x_{\xi i}(t)+10\zeta_i(t),\\
\xi_i(t)=x_{\xi i}(t),
\end{array}
\right.\quad
\left\{
\begin{array}{ll}
\dot{x}_{r i}(t)=\begin{bmatrix}
-1 & 1\\
-1 & -1
\end{bmatrix} x_{r i}(t)+10\begin{bmatrix}
1\\1
\end{bmatrix}\sum_{j=1,j\neq i}^3\zeta_j(t),\\
r_i(t)=\begin{bmatrix}
0.7 &-0.7
\end{bmatrix}x_{r i}(t), \quad i=1,2,3
\end{array}
\right.
\end{equation}

Both the uncertainties have the stable dynamic systems, and are admissible uncertainties according to Definition 3 and Definition 4.  Given $\gamma=1.36$, the gains of the global mode dependent controllers of the form (6) are obtained by  Theorem 1:
\begin{align*}
K_1(1)&=\begin{bmatrix}-34.3634 & -36.2178\end{bmatrix},
&K_1(2)&=\begin{bmatrix}-35.4979 & -34.9940\end{bmatrix},
&K_1(3)&=\begin{bmatrix}-20.6279 & -8.4957\end{bmatrix},\\
K_1(4)&=\begin{bmatrix}-17.5842 & -7.1341\end{bmatrix},
&K_2(1)&=\begin{bmatrix}-3.3655 & 8.6813\end{bmatrix},
&K_2(2)&=\begin{bmatrix}-4.9587 & 2.4603\end{bmatrix},\\
K_2(3)&=\begin{bmatrix}-3.4989 & 9.2715\end{bmatrix},
&K_2(4)&=\begin{bmatrix}-5.19097 & 1.7874\end{bmatrix},
&K_3(1)&=\begin{bmatrix}1.9156 & -7.0101\end{bmatrix},\\
K_3(2)&=\begin{bmatrix}-8.7706 & 2.4661\end{bmatrix},
&K_3(3)&=\begin{bmatrix}-8.3017 & 0.0576\end{bmatrix},
&K_3(4)&=\begin{bmatrix}1.4272 & -7.1653\end{bmatrix}.
\end{align*}

Using (8), the local mode dependent controllers are obtained,
 \begin{align*}
 K_1(1)&=\begin{bmatrix}-34.7055 &    -35.8488\end{bmatrix},
 &K_1(2)&=\begin{bmatrix}-18.5513 &    -7.5668\end{bmatrix},
 &K_2(1)&=\begin{bmatrix}-3.4099 & 8.8778\end{bmatrix},\\
 K_2(2)&=\begin{bmatrix}-5.1242 &     1.9806\end{bmatrix},
 &K_3(1)&=\begin{bmatrix}1.6630 & -7.0904\end{bmatrix},
 &K_3(2)&=\begin{bmatrix}-8.5192 & 1.1747\end{bmatrix}.
\end{align*}

Without controllers, the open-loop system is not stable, and the state trajectories are shown in Fig. 3.

\begin{center}
\begin{picture}(70,150)
\put(-60,0){\resizebox{70mm}{50mm}{\includegraphics[width=0.5\textwidth]{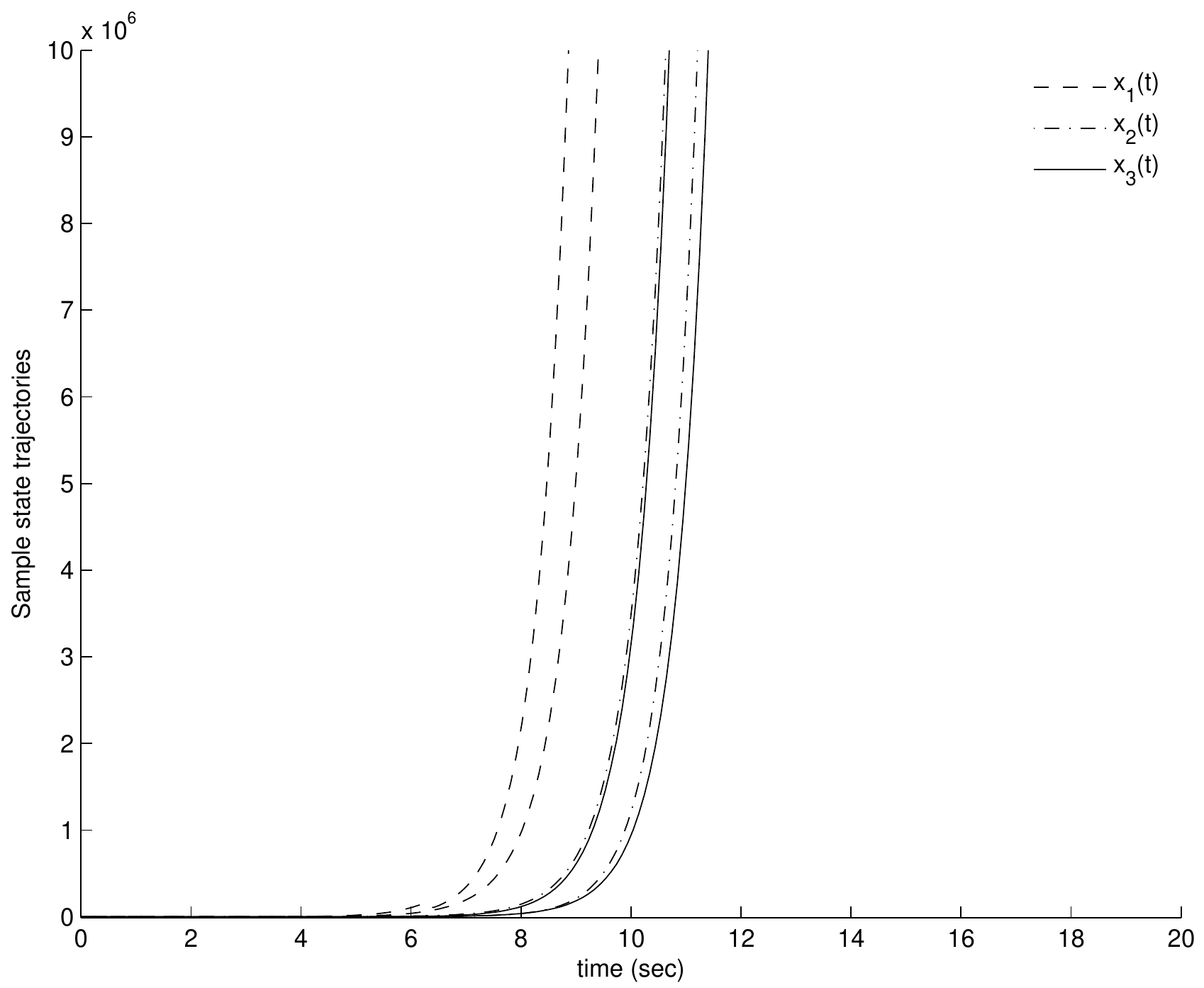}}}
\put(-100,-10){\shortstack{\quad  Fig. 3. Initial response  of the open-loop system without disturbance signals. }}
\end{picture}
\end{center}

Suppose the disturbance input signals are given by  $w_i(t)=e^{-0.5t},~i=1,2,3$, the initial state response of the closed-loop system is shown in Fig. 4. It shows the controllers can effectively stabilize the system.

\begin{center}
\begin{picture}(70,150)
\put(-60,0){\resizebox{70mm}{50mm}{\includegraphics[width=0.5\textwidth]{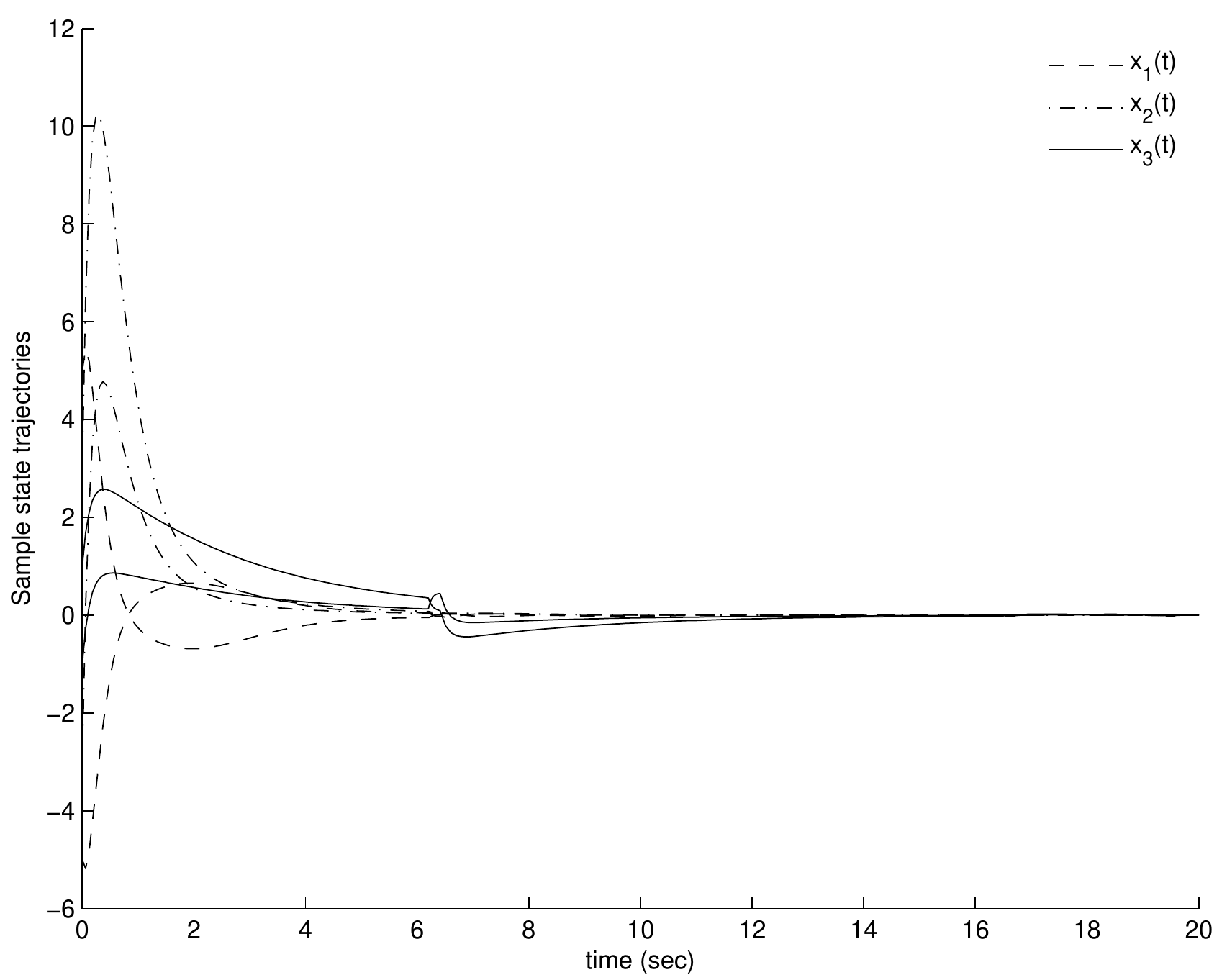}}}
\put(-100,-10){\shortstack{\quad  Fig. 4. Initial response  of the closed-loop system with disturbance signals. }}
\end{picture}
\end{center}

%

\section{Conclusions}

This paper has studied the decentralized state feedback $H_\infty$ control problem  for a class of uncertain Markovian jump large-scale systems. The proposed controllers are constructed with local system states and local operation modes of the subsystems to guarantee the entire system  stable with $H_{\infty}$-performance.  A sufficient condition in terms of rank constrained LMIs has been developed to construct such controllers. Finally, the developed theory has been illustrated by a numerical example.


\begin{thebibliography}{0}
\bibitem{Sandell1978}
Sandell, N.R., Varaiya, P., Athans, M., Safonov, M.G.: 'Survey of decentralized control methods for large scale systems', \emph{IEEE Trans. Autom. Control}, 1978, 23, (2), pp.108--128

\bibitem{Ugrinovskii2005}
Ugrinovskii, V.A., Pota, H.R.: 'Decentralized control of power systems via robust control of uncertain Markov jump parameter systems', \emph{Int. J. Control}, 2005, 78, (9), pp.662--667

\bibitem{Xiong2009}
Xiong, J., Ugrinovskii, V.A., Petersen, I.R.: 'Local Mode Dependent Decentralized Stabilization of Uncertain Markov Jump Large-Scale Systems', \emph{IEEE Trans. Autom. Control}, 2009, 54, (11), pp.2632--2637

\bibitem{Ugrinovskii2000}
Ugrinovskii, V.A., Petersen, I.R., Savkin, A.V., Ugrinoskaya, E.Y.: 'Decentralized state-feedback stabilization and robust control of uncertain large-scale systems with integrally constrained interconnections', \emph{Syst. Control Lett.}, 2000, 40, pp.107--119



\bibitem{Farias1992}
de Farias, D.P., Geromel, J.C., do Val, J.B.R., Costa, O.L.V.: 'Output feedback Control of Markov Jump Linear Systems in Continuous-Time', \emph{IEEE Trans. Autom. Control}, 2000, 45, (5), pp.944--949


\bibitem{Ma2009}
Ma, S., Boukas, E.K.: 'Stability and robust stabilisation for uncertain discrete stochastic hybrid singular systems with time delay
', \emph{IET Control Theory Appl}, 2009, 3, (9), pp.1217--1225

\bibitem{Feng1992}
Feng, X., Loparo, K.A., Ji, Y.,  Chizeck, H.J.: 'Stochastic Stability Properties of Jump Linear Systems', \emph{IEEE Trans. Autom. Control}, 1992, 37, (1), pp.38--52

\bibitem{xiong2009}
Xiong, J., Lam, J.: 'Robust $H_2$ control of Markovian jump systems with uncertain switching probabilities', \emph{Int. J. Syst. Sci.},  2009, 40, (3), pp.254--265



\bibitem{Fagoso1993}
de Souza, C.E., Fragoso,  M.D.: '$H_\infty$ Control for Linear Systems with Markovian Jumping Parameters', \emph{Control Theory and Advanced Technol.}, 1993, 9, (2), pp.457--466



\bibitem{Li2007}
Li, L.,  Ugrinovskii, V.A.: 'On Necessary and Sufficient Conditions for $H_\infty$ Output Feedback Control of Markov Jump Linear Systems', \emph{IEEE Trans. Autom. Control}, 2007, 52, (7), pp.1287--1292

\bibitem{Zhang2009}
Zhang, L., Boukas, E.K.: '$H_\infty$ control of a class of extended Markov jump linear systems', \emph{IET Control Theory Appl}, 2009, 3, (7), pp.834--842

\bibitem{Wu2012H2}
 Wu, L., Yao, X., Zheng, W.: 'Generalized '$H_{2}$ Fault Detection for Markovian Jumping
Two-Dimensional Systems', \emph{Automatica}, 2012, 48, (8), pp. 1741--1750

\bibitem{Wu2008}
Wu, L., Shi, P., Gao, H., Wang, C.: '$H_{\infty}$ Filtering for 2D Markovian Jump Systems', \emph{Automatica}, 2008, 44, (7), pp. 1849--1858

\bibitem{Wu2012L2}
 Wu, L., Su, X., Shi, P.: 'Sliding Mode Control with Bounded $L_{2}$ Gain Performance of Markovian
Jump Singular Time-Delay Systems', \emph{Automatica}, 2012, 48, (8), pp.
1929--1933

\bibitem{Wu2010}
 Wu, L.,Shi, P.,Gao, H.: ' State Estimation and Sliding-Mode Control of Markovian Jump
Singular Systems', \emph{IEEE Trans. Autom. Control}, 2010, 55, (5), pp. 1213--1219

\bibitem{Zhang2008}
Zhang, L., Huang, B., Lam, J.: '$H_{\infty}$ model reduction of Markovian jump
linear systems', \emph{Syst. Control Lett.}, 2003, 50, (2), pp.103--118

\bibitem{Xiong2010}
Xiong, J., Ugrinovskii, V.A., Petersen, I.R.: 'Decentralized Output Feedback Guaranteed Cost Control of Uncertain Markovian Jump Large-Scale Systems: Local Mode Dependent Control Approach', In: K.~Grigoriadis, M.~Mohammadpour (Eds), \emph{Efficient Modeling and Control of Large-scale Systems} (Springer, 2010), pp.167--196



\bibitem{li2007}
Li, L., Ugrinovskii, V.A., Pota, H.R.: 'Decentralized Robust Control of Uncertain Markov Jump Parameter Systems via Output Feedback', \emph{Automatica}, 2007, 43, (11), pp.1932--1944

\bibitem{Petersen2006}
Petersen, I.R.: 'Decentralized state feedback guaranteed cost control of uncertain systems with uncertainty described by integral quadratic constraints',  \emph{Proc. American Control Conf.,
Minneapolis, Minnesota, USA}, 2006, pp.1333--1339


\bibitem{Marcos2008}
Todorov, M.G., Fragoso,  M.D.: 'Infinite Markov jump bounded real lemma', \emph{Syst. Control Lett.},  2008, 57, pp. 64--70

\bibitem{li20072}
Li, L., Petersen, I.R.: 'A rank constrained LMI algorithm for decentralized state feedback guaranteed cost control of uncertain systems with uncertainty described by integral quadratic constraints',\emph{Proc. American Control Conf.,
New York City, USA}, 2007, pp.796--801



\bibitem{Iwasaki1995}
Iwasaki, T., Skelton, R.E.:'The XY-centering algorithm for the dual LMI problem: A new
approach to ﬁxed order control desgin', \emph{Int. J. Contr.},  1995, 62, (6), pp.1257--1272

\bibitem{Ghaoui1997}
Ghaoui, L. E, Oustry, F. Rami, M.A.:'Acone complementarity linearization algorithm for static output-feedback and related problems', \emph{IEEE Trans. Autom. Control}, 1997, 42, (8),
 pp.1171--1176

\bibitem{Orsi2006}
Orsi, R., Helmke, U., Moore, J.B.: 'A Newton-like method for solving rank constrained linear matrix inequalities',\emph{Automatica}, 2006, 43, (11), pp.1875--1882

\bibitem{Orsi}
Orsi, R.: 'LMIRank: software for rank constrained LMI problems,
2005', Available from http://rsise.anu.edu.au/∼robert/lmirank/



\end{thebibliography}
\end{document}